\documentclass[12pt,journal,final,letterpaper,onecolumn]{article}
\usepackage[top=1in, bottom=1in, left=1in, right=1in]{geometry}

\usepackage{graphicx} 
\usepackage{subfigure}
\usepackage{algorithm}
\usepackage{algorithmic}
\usepackage{hyperref}
\usepackage{amsmath,  amssymb, amsthm}
\usepackage[comma,authoryear]{natbib}

\usepackage{color}

\usepackage{setspace}
\newcommand{\beq}{\vspace{0mm}\begin{equation}}
\newcommand{\eeq}{\vspace{0mm}\end{equation}}
\newcommand{\beqs}{\vspace{0mm}\begin{eqnarray}}
\newcommand{\eeqs}{\vspace{0mm}\end{eqnarray}}
\newcommand{\barr}{\begin{array}}
\newcommand{\earr}{\end{array}}

\def\d{\mbox{d}}

\newcommand{\zv}{\boldsymbol{z}}

\newcommand{\thetav}{\boldsymbol{\theta}}

\newcommand{\E}{\mathbb{E}}

\newtheorem{thm}{Theorem} 
\newtheorem{cor}[thm]{Corollary}

\usepackage{natbib,natbibspacing}
\renewcommand\footnotemark{}
\begin{document}
\title{Sample Size Dependent Species Models} 
\author{Mingyuan~Zhou
and Stephen~G~Walker
\\ 
The University of Texas at Austin
\thanks{ 
M. Zhou is with the Department of Information, Risk, and Operations Management, McCombs School of Business, 
and S. G. Walker is with the Departments of Mathematics and  Statistics \& Data Sciences, 
the University of Texas at Austin, Austin, TX 78712, USA. \emph{Emails:} \texttt{mingyuan.zhou@mccombs.utexas.edu, s.g.walker@math.utexas.edu}.
}
}

\maketitle

\vspace{-3mm}

\begin{spacing}{1.5}
\begin{abstract}

Motivated by the fundamental problem of measuring species diversity, this paper introduces the concept of a cluster structure to define an exchangeable cluster probability function that governs the  joint distribution of a random count  and its exchangeable random partitions. A cluster structure, naturally arising from a completely random measure mixed Poisson process, allows the probability distribution of the random partitions of a subset of a sample to be dependent on the sample size, a distinct and motivated feature that differs it from  a partition structure. A generalized negative binomial process  model is proposed to generate a cluster structure, where in the prior the number of clusters is finite and Poisson distributed, and the cluster sizes follow a truncated negative binomial distribution. We construct a nonparametric Bayesian estimator of Simpson's index of diversity under the generalized negative binomial process. We illustrate our results through the analysis of two real sequencing count datasets.  
 
 \vspace{5mm}
\end{abstract}
\end{spacing}

\newpage
\section{Introduction}

A fundamental problem in biological and ecological studies is to measure the degree of diversity of a population whose individuals are classified into different groups; see \cite{Fisher1943, simpson1949measurement,hill1973diversity} and \cite{magurran2004measuring}.
The rapid development of modern sequencing technologies  also generates significant recent interest in the measurement of population diversity  using samples 
 summarized as the frequencies of observed sequences
\citep{hughes2001counting,
shaw2008s,bunge2014estimating,guindani2014bayesian}. 
The Simpson's index of diversity, widely used to measure species evenness,  is defined as the probability for two individuals randomly selected  from a population to be from different groups  \citep{simpson1949measurement}. 
Thus, if $\pi_k$ denotes the population probability for an individual to be in group $k$, with 
$\sum_{k\geq 1}\pi_k=1$, then the Simpson's index of diversity is defined as 
\begin{equation}
S=1-\sum_{k=1}^K\pi_k^2\label{simp}
\end{equation}
which is also understood to be $P(z_1\ne z_2)$, where $z_i$ is the group individual $i$ is assigned to. Here, $K$ could be finite or infinite though \citet{simpson1949measurement} assumed it to be finite. 

A sample estimate for (\ref{simp}), which is unbiased, is given by
\beq
\widehat{S}=1- \sum_{k=1}^K \frac{n_k(n_{k}-1)}{n(n-1)}, \label{eq:S_hat}
\eeq
where 
$$n_k=\sum_{i=1}^n {\bf 1}(z_i=k).$$ 

Alongside Simpson's index of diversity, other diversity indices have been proposed to measure species richness; see \citet{bunge1993estimating,chao2005species} and \citet{bunge2014estimating} for reviews. Recent nonparametric Bayesian approaches to species diversity,  focusing on the study of species richness, derive  the distribution of the number of new species via  $n'$ new individuals randomly selected from the population, given a sample of size $n$; see \cite{lijoi2007bayesian,lijoi2008bayesian} and \cite{favaro2009bayesian,favaro2013conditional}. These papers form the basis for  Bayes nonparametric estimators of the Simpson's index of diversity,  as in \citet{cerquetti2012bayesian}. 

The underlying structure of the Bayesian species sampling models  are built on Kingman's concept of a partition structure, \citep{kingman1978random,kingman1978representation}, which defines a family of 
consistent probability distributions for random partitions of a set $[m]:=\{1,\ldots,m\}$. The sampling consistency requires  the  probability distribution of the random partitions of a subset of size $m$ of a set of size $n\geq m$ to be the same for all $n$.
More specifically, for a random partition $\Pi_m=\{A_1,\ldots,A_l\}$ of the set $[m]$, 
where there are $l$ clusters and each element $i\in[m]$ belongs to one and only one set $A_k$ from $\Pi_m$, such a constraint requires
that $P(\Pi_m|n)=P(\Pi_m|m)$ does not depend on $n$, where $P(\Pi_m|n)$ denotes the marginal partition probability for $[m]$ when it is known the sample size is $n$. 
As further developed in \citep{pitman1995exchangeable,csp}, if 
$P(\Pi_m|m)$ depends only on the number and sizes of the $(A_k)$, regardless of their order, then it is called an exchangeable partition probability function (EPPF) of $\Pi_m$, expressed as  $P(\Pi_m=\{A_1,\ldots,A_l\}|m)=p_m(n_1,\ldots,n_l)$, where $n_k=|A_k|$. 
The sampling consistency  
amounts to an addition rule \citep{csp,Gnedin_deletion} for %
the EPPF; that $p_1(1) = 1$ and
\beqs\label{eq:addrule}
p_m(n_1,\ldots,n_l) =  p_{m+1}(n_1,\ldots,n_l,1)+ \sum_{k=1}^l p_{m+1}(n_1,\ldots,n_k+1,\ldots,n_l).
\eeqs
An EPPF of  $\Pi_m$ satisfying this constraint is considered as an EPPF of $\Pi :=(\Pi_1,\Pi_2,\ldots)$. 
For an EPPF of $\Pi$, 
 $\Pi_{m+1}$ can be constructed from  $\Pi_m$ by assigning element $(m+1)$ to $A_{z_{m+1}}$ based on the prediction rule as
\beq\notag
z_{m+1}|\Pi_m=
\begin{cases} \vspace{3mm}
l+1& \mbox{with probability  }\frac{p_{m+1}(n_1,\ldots,n_l,1)}{p_m(n_1,\ldots,n_l)} , \\ 
k & \mbox{with probability  }\frac{p_{m+1}(n_1,\ldots,n_k+1,\ldots,n_l)}{p_m(n_1,\ldots,n_l)}.\end{cases}
\eeq

A basic EPPF of $\Pi$ is the Ewens sampling formula \citep{ewens1972sampling,Antoniak74}. Moving beyond the Ewens sampling formula, 
various  approaches, including the Pitman-Yor process 
 \citep{perman1992size,pitman1997two}, Poisson-Kingman models \citep{pitman2003poisson}, species sampling \citep{Pitman96somedevelopments}, 
  stick-breaking priors \citep{ishwaran2001gibbs}, and Gibbs-type random partitions \citep{gnedin2006exchangeable},   have been proposed to construct more general EPPFs of $\Pi$.  
See \citet{muller2004nonparametric}, \citet{BeyondDP} and \citet{Muller2013} for reviews. 
Among these approaches,
there has been increasing interest in 
normalized random measures with independent increments (NRMIs)
\citep{regazzini2003distributional}, 
where a completely random measure \citep{Kingman,PoissonP} with a finite and strictly positive total random mass is normalized to construct a random probability measure. For example, the normalized gamma process is a Dirichlet process \citep{ferguson73}. 
More advanced completely random measures, such as the generalized gamma process of \citet{brix1999generalized}, can be employed to produce more general exchangeable random partitions  of $\Pi$ \citep{pitman2003poisson,csp,lijoi2007controlling}. However,  the expressions of the EPPF and its associated prediction rule usually involve integrations that are difficult to calculate. 

With respect to the Simpson's measure of diversity, it is our contention that a prior model for this quantity; i.e. $P(z_1\ne z_2)$ should depend on $n$ and hence we write $P(z_1\ne z_2|n)$ meaning that in general, rather than the marginal distribution of $(z_1,\ldots,z_m)$, with $(z_{m+1},\ldots,z_n)$ integrated out, being independent of the sample size $n\geq m$, it actually does depend on $n$. 

The motivation for this is that as $n$ increases, so could the possible groups which are available for classification. It is anticipated that unknown species emerge, which is different from known species first being seen, as samples are collected. Hence, the probability, according to the experimenter's prior model, that $z_1$ and $z_2$ belong to the same group will, for example, diminish with $n$ if, as the sample size increases, it is thought more appropriate for individuals to be reclassified into different species.
In short, if all the possible species are known upfront then it is possible to classify $z_1$ and $z_2$ once and for all having seen just them. However, if there is uncertainty about the species, even whether $z_1$ and $z_2$ are the same species or not, which in life is often reality, then reassessing their classifications with $n$ should occur and hence a model for which $P(z_1\ne z_2)$ changes with $n$ is motivated.

Consequently, in a Bayesian context, we will be facilitating the dependence of  $(z_1,\ldots,z_m)$, for all $m\leq n$, on $n$. To develop this theme, and to allow the mathematics to proceed in a neat way, and without forcing any restrictions, we make $n$ a random object within the model.

We work at a fundamental level with a normalized completely  random measure. Hence, the total (random) mass  is unidentified and consequently arbitrary. We take this opportunity to use it to model the, prior to observation, random sample size $n$.  More specifically, we model the sample size $n$ as a Poisson random variable 
the  mean of which is parameterized by the total random mass of a completely random measure $G$ over a complete and separable metric space $\Omega$. The total random mass $G(\Omega)$ 
is used to normalize $G$ to obtain a random probability measure $G(\cdot)/G(\Omega)$.
Linking $n$ to $G(\Omega)$ with a Poisson distribution 
makes the scale of $G$ become identifiable.
 With $G$ marginalized out, the joint distribution of  $n$ and its exchangeable random partition $\Pi_n$ is called an exchangeable cluster probability function (ECPF). 
On observing a sample of size $n$, we are interested in the EPPF $P(\Pi_n|n)$ and marginalizing over $n-m$ elements we would consider $P(\Pi_m|n)$. Note that distinct from a partition structure, we no longer require $P(\Pi_m|n)=P(\Pi_m|m)$ for $n>m$ in  a cluster structure.

Specifically, we consider a generalized negative binomial (NB) process model 
where
$G$ is drawn from a generalized gamma process of \citet{brix1999generalized}. 
A draw from the generalized NB process (gNBP) represents a cluster structure with a Poisson distributed finite number of clusters, whose sizes 
follow a truncated NB distribution. Marginally, the sample size 
follows a generalized NB distribution.  
These three count distributions and the prediction rule are determined by a discount, a probability and a mass parameter. These parameters are convenient to infer using the fully factorized ECPF. Since $P(\Pi_m|n)= P(\Pi_m|m)$ is often not true for $n>m$,  the EPPF of the gNBP, which is derived by applying Bayes' rule on the ECPF and the generalized NB distribution, generally violates the addition rule and hence is dependent on the sample size.  This EPPF will be  referred as the generalized Chinese restaurant sampling formula. To generate an exchangeable random partition of $[n]$ under this EPPF, we show we could use either a Gibbs sampler or a recursively-calculated sequential prediction rule.

The layout of the paper is as follows: In Section 2 we provide all the necessary preliminary notation and a description of normalized random measures, while in Section 3 we introduce the new model for constructing sample size dependent species models. In Section 4 we apply the theory in Section 3 to the generalized negative binomial process and we present real data applications in Section 5. We end the paper with a brief conclusion and provide   the proofs of theorems and corollaries in the Appendix.

\section{Preliminaries}\label{sec:preliminary}

In this section we provide the mathematical foundations for an independent increment process with no Gaussian component.  These are pure jump processes and for us will have finite limits so that the process can be normalized by the total sum of the jumps to provide a random distribution function.
The most well known of such processes is the gamma process (see, for example, \citet{ferguson1972representation}) and we will be specifically working with a generalized gamma process in Section 2.1.

\subsection{Generalized Gamma Process}

The generalized gamma process,  which we will denote by $\mbox{g}\Gamma\mbox{P}(G_0,a,1/c)$, is a completely random (independent increment) measure defined on the product space $\mathbb{R}_+\times \Omega$, where $a< 1$ is a discount parameter and $1/c$ is a scale parameter \citep{brix1999generalized}. 
It assigns independent infinitely divisible generalized gamma random variables $G(A_j)\sim{{}}\mbox{gGamma}(G_0(A_j),a,1/c)$ to disjoint Borel sets $A_j\subset \Omega$, 
with Laplace transform given by
\beq\label{eq:Laplace}
\E\left[e^{-\phi\,G(A)}\right] = \exp\left\{-\frac{G_0(A)}{a}\left[(c+\phi)^a-c^a\right]\right\}.
\eeq
The L\'{e}vy measure of the generalized gamma process
can be expressed as
\beqs\label{eq:LevyGGP}
\nu(\d s \,,\d\omega) = \frac{1}{\Gamma(1-a)}r^{-a-1}e^{-cr}\,\d s \,G_0(\d\omega).
\eeqs
The connection between (\ref{eq:Laplace}) and (\ref{eq:LevyGGP}), not given here, is the well known form for the Laplace transform of an infinitely divisible random variable.

When $a\rightarrow0$, we recover the gamma process,   and if $a=1/2$, we recover the inverse Gaussian process \citep{lijoi2005inverseGaussian}. 
 A draw $G$ from 
$\mbox{g}\Gamma\mbox{P}(G_0,a,1/c)$ 
 can be expressed as
\beq
G = \sum_{k=1}^{K} r_k \delta_{\omega_k},\notag \eeq
with $K\sim\mbox{Po}(\nu^+)$ and $(r_k,\omega_k)\stackrel{i.i.d.}{\sim} \pi(\d s\,,\d\omega)$, 
where $r_k=G(\omega_k)$ is the weight for atom $\omega_k$ 
and $\pi(\d s\, ,\d\omega)\nu^{+} \equiv \nu(\d s\,,\d\omega)$. 
Except where otherwise specified, we only 
consider $a<1$ and $c>0$. 
If $0\le a<1$, since the Poisson intensity $\nu^+ = \nu(\mathbb{R}_+\times \Omega) = \infty$ (i.e., $K=\infty$ a.s.) and
  $
 \int_{\mathbb{R}_+\times \Omega} \min\{1, s\} \nu(\d s\, \d\omega)  
 $
 is finite, a draw from $\mbox{g}\Gamma\mbox{P}(G_0,a,1/c)$ consists of countably infinite atoms. On the other hand, if $a<0$, then $\nu^+=-\gamma_0c^a/a$ and thus $K\sim \mbox{Po}(-\gamma_0c^a/a)$ (i.e., $K$ is finite a.s.) and $r_k\stackrel{i.i.d.}{\sim}\mbox{Gamma}(-a,1/c)$. This process will be seen again in Section 4.

\subsection{Normalized Random Measures 
}\label{NRMI}

A NRMI model \citep{regazzini2003distributional} is
a normalized completely random measure 
$$\widetilde{G}=G/G(\Omega)$$ where 
$G(\Omega)=\sum_{k=1}^{K} r_{k}$ is the total random mass, which is required to be finite and strictly positive. 
Note that the strict positivity of $G(\Omega)$ implies that $\nu^+=\infty$ and hence $K=\infty$  a.s. \citep{regazzini2003distributional,BeyondDP}. For us we will not necessarily be assuming that $K=\infty$ a.s. In fact our model is such that $K=0\iff n=0$, which is coherent, and, moreover, $P(K=0|n>0)=0$.  

Here we describe how the random allocations of individuals to groups are distributed based on the independent random jumps of the generalized gamma process. 
With a random draw $G = \sum_{k=1}^{K} r_k \delta_{\omega_k}$, by introducing a categorical latent  variable $z$ with
$
P(z=k|G) =r_k/G(\Omega), \notag 
$ 
when a sample of size $n$ is observed  
we have
\beqs\label{eq:f_G_N}
p(\zv|G,n)= \prod_{i=1}^n \frac{r_{z_i}}{\sum_{k=1}^{K} r_k}
=  \left(\sum_{k=1}^{K} r_k\right)^{-n}
\prod_{k=1}^{K} {r_k^{n_k}} , 
\eeqs
where  $\zv=(z_1,\ldots,z_n)$ is a sequence of categorical random variables  indicating the cluster memberships, $n_k = \sum_{i=1}^n  \mathbf{1}(z_{i}=k)$ is the number of data points assigned to category $k$, and $n=\sum_{k=1}^{K} n_k$.  A random partition $\Pi_n$ of $[n]$ is defined by the ties between the $(z_i)$. So at this point, (\ref{eq:f_G_N}) is standard.

Now (\ref{eq:f_G_N}) exhibits a lack of identifiabilty in that the scale of the $(r_k)$ is arbitrary; the model is the same if we set $\widetilde{r}_k=\kappa\,r_k$ for any $\kappa>0$. Hence, the total mass 
$\sum_{k=1}^K r_k$ is unidentified. 

Additionally, for reasons outlined in Section 1, we want, having marginalized out $G$, for $n$ to remain, and for us to have $p(\zv|n)$ to remain. For the standard models, when $G$ is integrated out, $n$ disappears and we have $p(\zv)$ depending solely on the parameters  of the model.

We solve both these issues by allowing $n$ to depend on $G$ via 
$$p(n|G)=\mbox{Po}[G(\Omega)],$$
from which we have independently
$$p(n_k|G)=\mbox{Po}(r_k).$$
We note here then that the prior model is for $p(n,G)$ and, consequently, $p(G|n)$ means $G$ depends on $n$; i.e. for each $n$ we will have a different random measure for $G$.

We provide in Section 3 the general form for the prior $p(\zv|n)$ and in Section 4 the specific case when $G$ is a generalized gamma process. 
In Section 5 we use MCMC methods to estimate the posterior values of Simpson's index of diversity using real sequence frequency count data. 

Posterior inference via MCMC is also simplified by our approach.
Following \citet{
james2009posterior}, 
 a specific auxiliary variable  $T>0$, with 
 $ 
p_T(t|n,G(\Omega)) =\mbox{Gamma}(n,1/G(\Omega))
$, 
can be introduced to yield a fully factorized likelihood, stimulating the development of a number of posterior simulation algorithms including 
\citet{griffin2011posterior,barrios2012modeling} and \citet{
 favaromcmc}.  
Marginalizing out $G$ and then $T$ from that fully factorized likelihood leads to an EPPF of $\Pi$ \citep{pitman2003poisson,csp,lijoi2007controlling}.  However, the prediction rule of the EPPF may not be easy to calculate.

\section{Structure of Model}\label{sec:CountMixture}
As has been previously mentioned,  we  link the sample size $n$ to the total random mass of $G$ with a Poisson distribution; 
\beq
p(n| G)=\mbox{Po}\big[ G(\Omega)\big].\label{pois}
\eeq 
Since the $n$ data points are clustered according to the normalized random probability measure $G/G(\Omega)$, we have the equivalent sampling mechanism given by
\beq
p(n_k|G)=\mbox{Po}( r_k)\quad\mbox{independently for}\quad k=1,2,\ldots\,, \notag
\eeq 
and, since $n=\sum_k n_k$, we obviously recover (\ref{pois}).

Therefore, we link directly the cluster sizes $(n_k)$ to the weights $(r_k)$ with  independent Poisson distributions, which is in itself an appealing intuitive feature. 
The mechanism to generate a sample of arbitrary size is now well defined and $G$ is no longer scaled freely.  The new construction also allows $G(\Omega)=0$, for which $n\equiv 0$ a.s.  Allowing $G(\Omega)=0$ with a nonzero probability relaxes the requirement of $\nu^+=\infty$ (i.e., $K=\infty$ a.s.). 

 \begin{figure}[!tb]
\begin{center}
\includegraphics[width=0.96\columnwidth]{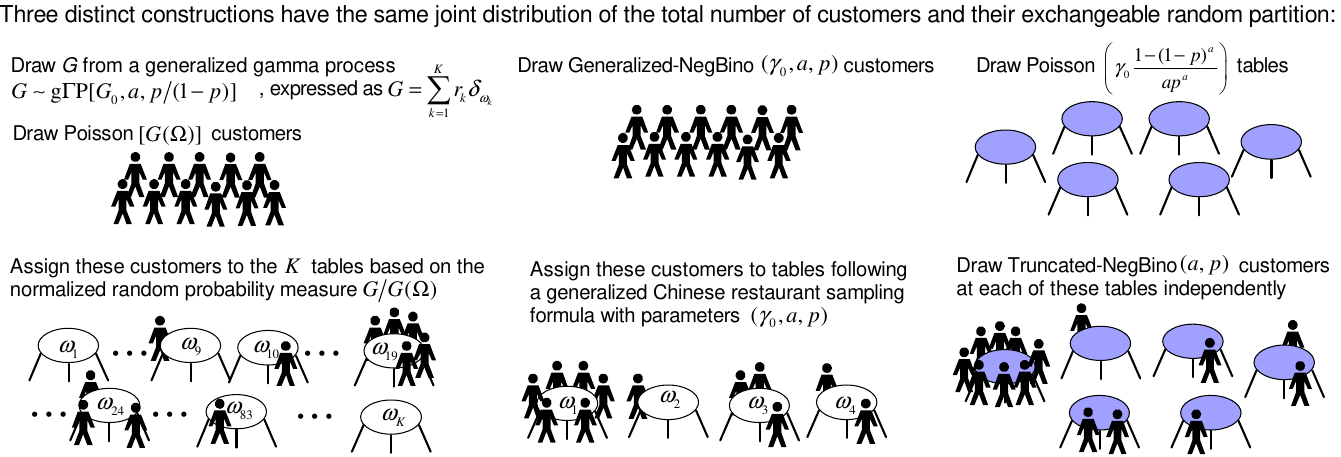}
\end{center}
\vspace{-6mm}
\caption{ \label{fig:gNBPdraw} \small The cluster structure of the generalized NB process 
 can be either constructed by assigning $\mbox{Pois}[G(\Omega)]$ number of customers to tables following a normalized generalized gamma process $G/G(\Omega)$, where $G\sim\mbox{g}\Gamma\mbox{P}[G_0,a,p/(1-p)]$, or constructed by assigning $n\sim\mbox{gNB}(\gamma_0,a,p)$ number of customers to tables following a generalized Chinese restaurant sampling formula $\zv\sim$~$\mbox{gCRSF}(n,\gamma_0,a,p)$, where $\gamma_0=G_0(\Omega)$. A equivalent cluster structure can also be generated by first drawing $\mbox{Pois}\big(\gamma_0\frac{1-(1-p)^a}{ap^a}\big)$ number of tables, and then drawing  $\mbox{TNB}(a,p)$ number of customers independently at each table. 
}
\end{figure}

A key insight of this paper is that 
a completely random measure mixed Poisson process produces a cluster structure that is identical in distribution to both (i) the one produced  
by assigning the total random count of the Poisson process into exchangeable random partitions, using  the random probability measure normalized from that completely random measure, and (ii) the one produced by assigning the total (marginal) random count $n$ of the mixed Poisson process into exchangeable random partitions using an EPPF of $\Pi_n$.  For example, when the generalized gamma process \citep{brix1999generalized} is used as the completely random measure in this setting, our key discoveries are summarized in 
Figure \ref{fig:gNBPdraw}, 
which will be discussed further in Section 4.

We note that  \citet{BNBP_PFA_AISTATS2012} and \citet{NBP2012} have explored related ideas to mix a gamma or beta process with a negative binomial process, and use that hierarchical process for mixture modeling of grouped data. Yet the authors marginalized neither  the beta nor gamma process due to technical difficulties and relied on finite truncation for inference. We will discuss at the end of the paper that the ideas and techniques developed in this paper serve as the foundation for the authors to develop priors for random count matrices and understand the marginal combinatorial structures of the beta-negative binomial process.

In the following theorem, we establish  the marginal model for the $(n_k)$ with $G$ marginalized out. The proof for this theorem is provided in the Appendix. 

\begin{thm}[Compound Poisson Process]\label{thm:compoundPoisson}
 It is that the $G$ mixed Poisson process is also a compound Poisson process; a random draw of which can be expressed as 
$$X(\cdot)=\sum_{k=1}^{l} n_k \,\delta_{\omega_k}(\cdot)\quad\mbox{with }l\sim\emph{\mbox{Po}}\left[G_0(\Omega)\int_{0}^\infty(1-e^{-s})\rho(\d s)\right],$$
and independently
$$P(n_k=j)=\frac{ {\int_{0}^\infty s^j e^{-s} \rho(\d s)} }{{j!}\int_{0}^\infty(1-e^{-s})\rho(\d s)}~~\mbox{for}~~j=1,2,\ldots\notag$$
 where $\int_{0}^\infty(1-e^{-s})\rho(\d s) 
<\infty$ is a condition required for the characteristic functions of $G$ to be well defined, $\omega_k\stackrel{iid}{\sim} g_0$ and $g_0(\d\omega)=G_0(\d\omega)/G_0(\Omega)$. 
 \end{thm}

The compound Poisson representation dictates 
the model to have a Poisson distributed finite number of  clusters, whose sizes follow a positive discrete distribution. 
 The mass parameter $\gamma_0=G_0(\Omega)$ has a linear relationship with the expected number of clusters, 
 but has no direct impact on the cluster-size distribution.  
 Note that a draw from $G$ contains $K< \infty$ or $K=\infty$ atoms a.s., but only $l$ of them would be associated with nonzero counts if $G$ is mixed with a Poisson process. 
Since the cluster indices are unordered and exchangeable, 
without loss of generality, in the following discussion,  we relabel 
the atoms with nonzero counts 
in order of appearance from $1$ to $l$  
and then $z_i\in\{1,\ldots,l\}$ for $i=1,\ldots,n$,  with $n_k>0$ if and only if $1\le k \le l$ and $n_k=0$ if $k>l$. 

\begin{cor}[Exchangeable Cluster/Partition Probability Functions]\label{thm:ECPF}
The model 
has a fully factorized  exchangeable cluster probability function (ECPF) as
$$p(\zv,n|\gamma_0,\rho) = 
 \frac{\gamma_0^l}{n!} \exp\left\{\gamma_0\int_{0}^\infty(e^{-s}-1)\rho(\d s)\right\}
\prod_{k=1}^l \int_0^\infty s^{n_k} e^{-s} \rho(\d s),$$
the marginal distribution for the sample size $n=X(\Omega)$ has probability generating function
$$\E[t^{n}|\gamma_0,\rho] = \exp\left\{ \gamma_0 \int_{0}^\infty (e^{-(1-t)s}-1)\rho(\d s)\right\}$$ 
 and probability mass function 
$$\left. p_N(n|\gamma_0,\rho)=\frac{d^n (\E[t^n|\gamma_0,\rho])} {d t^n } \right |_{t=0},$$
and an exchangeable partition probability function (EPPF) of $\Pi_n$ as
$$p(\zv|n,\gamma_0,\rho) = {p(\zv,n|\gamma_0,\rho) }\big/{p_N(n|\gamma_0,\rho)}.\notag$$ 
\end{cor}

The proof of this is straightforward given the representation in Theorem \ref{thm:compoundPoisson} and given the one-to-many-mapping combinatorial coefficient taking $(n_1,\ldots,n_l,l)$ to $(z_1,\ldots,z_n,n)$ is
$$\frac{l!}{n!}\,\prod_{k=1}^l n_k!\,\,.$$ 

\begin{cor}[Prediction Rule]\label{thm:predict} 
Let $l^{-i}$ represent the number of clusters in $\zv^{-i}:=\zv\backslash z_i$ and $n_k^{-i}:=\sum_{j\neq i} {\bf 1}(z_j=k)$. We can express the prediction rule of the  model as
\beq
P(z_{i} = k|\zv^{-i},n,\gamma_0,\rho) \propto
\begin{cases}\vspace{2mm}
\frac{\int_{0}^\infty s^{n_k^{-i}+1} e^{-s} \rho(\d s)}{\int_0^\infty s^{n_k^{-i}}e^{-s}  \rho(\d s)} , & \emph{\mbox{for }} k=1,\ldots,l^{-i};\\
 \gamma_0\int_0^\infty se^{-s}\rho(\d s), & \emph{\mbox{if } }k=l^{-i}+1.
\end{cases}\notag
\eeq
This prediction rule can be used to simulate an exchangeable random partition of $[n]$ via Gibbs sampling.
\end{cor}

The proof for this Corollary is provided in the Appendix. In the next section we will study a particular process: the generalized negative binomial process, whose ECPF has a simple analytic expression and whose exchangeable random partitions can not only be simulated via Gibbs sampling using the above prediction rule, but also be sequentially constructed using a recursively calculated prediction rule. 

\section{Generalized Negative Binomial Process}\label{sec:gNBP} 
In the following discussion, 
we study the generalized NB process (gNBP) model where $G\sim\mbox{g}\Gamma\mbox{P}[G_0,a,p/(1-p)]$  with $a<0$, $a=0$ or $0<a<1$. 
Here we apply the results in Section 3 to this specific case.
Using (\ref{eq:LevyGGP}), we have
$$\int_{0}^\infty s^{n}e^{-s}\rho(\d s)= {\frac{\Gamma(n-a)}{{\Gamma(1-a)} }p^{n-a}}\quad \mbox{and} \quad\int_0^\infty(1-e^{-s})\rho(\d s) = \frac{1-(1-p)^a}{ap^{a}}.$$
Marginalizing out $\lambda$ from 
$n| \lambda\sim\mbox{Po}(\lambda)$ with $\lambda\sim{{}}\mbox{gGamma}[\gamma_0,a,p/(1-p)]$, leads to a generalized NB distribution; i.e. $n\sim\mbox{gNB}(\gamma_0,a,p)$, with shape parameter $\gamma_0$, discount parameter $a<1$, and probability parameter $p$. 
Denote by $\sum_{*}$ as the summation over all sets of positive integers $(n_1,\ldots,n_l)$ with ${\sum_{k=1}^l n_k = n}$.  As derived in the Appendix, the probability mass function (PMF) of the generalized NB distribution can be expressed as
\beq\label{eq:f_M0}
p_N(n|\gamma_0,a,p)
 = 
\frac{p^n}{n!}e^{-{\gamma_0}\frac{1-(1-p)^a}{ap^a}} \sum_{l=0}^n \gamma_0^l p^{-al} S_a(n,l),
\eeq 
where 
\begin{align}\label{eq:gStirling}
S_a(n,l) = \frac{n!}{l!}\sum_{*} \prod_{k=1}^l \frac{\Gamma(n_k-a)}{n_k!\Gamma(1-a)}= \frac{1}{l!a^{l}}\sum_{k=0}^l (-1)^k \binom{l}{k}  \frac{\Gamma(n-ak)}{\Gamma(-ak)}
\end{align}
 are 
generalized Stirling numbers of the first kind \citep{charalambides2005combinatorial,csp}, which can be recursively calculated via $S_a(n,1)={\Gamma(n-a)}/{\Gamma(1-a)}$, $S_a(n,n)=1$ and $S_a(n+1,l) = (n-al)S_a(n,l)+S_a(n,l-1)$. Note that when $-ak$ is a nonnegative integer,  $\Gamma(-ak)$ is not well defined but $\Gamma(n-ak)/\Gamma(-ak)=\prod_{i=0}^{n-1}(i-ak)$ is still well defined.

Marginalizing out $G$ in the generalized gamma process mixed Poisson process 
\beq\label{eq:gGaPP0}
X|G\sim\mbox{PP}(G)\quad\mbox{and}\quad G\sim{{}}\mbox{g}\Gamma\mbox{P}\left[G_0,a, {p}/{(1-p)}\right]
\eeq 
leads to a generalized NB process
$
X\sim\mbox{gNBP}(G_0,a,p),
$
such that for each $A\subset \Omega$, $X(A)\sim\mbox{gNB}(G_0(A),a,p)$. This process is also a compound Poisson process as
$$ 
X(\cdot)=\sum_{k=1}^{l} n_k\delta_{\omega_k}(\cdot),~l\sim\mbox{Po}\Big(\gamma_0\frac{1-(1-p)^a}{ap^a}\Big),~n_k  \stackrel{iid}{\sim} \mbox{TNB}(a,p),~\omega_k \stackrel{iid}{\sim} g_0,$$ where $\mbox{TNB}(a,p)$ denotes a truncated NB distribution, with PMF
\beqs\label{eq:TNB}
p_U(u|a,p)= \frac{\Gamma(u-a)}{u!\Gamma(-a)}\frac{p^u(1-p)^{-a}}{1-(1-p)^{-a}},~u=1,2,\ldots.
\eeqs

The ECPF of the gNBP 
model is given by
\beqs \label{eq:f_Z_M}
p(\zv,n|\gamma_0,a,p) 
=\frac{1}{n!}e^{-\gamma_0\frac{1-(1-p)^a}{ap^{a}}}
\gamma_0^{l{}} p^{n-al{}}
\prod_{k=1}^{l{}} \frac{\Gamma(n_k-a)}{{\Gamma(1-a)} }.
\eeqs 
The EPPF of $\Pi_n$ 
is the ECPF in (\ref{eq:f_Z_M}) divided by the marginal distribution of $n$ in (\ref{eq:f_M0}), given by
\begin{align} \label{eq:EPPF}
p(\zv|n,\gamma_0,a,p) = 
\frac{\gamma_0^{l}  p^{-al}}{
\sum_{\ell=0}^n \gamma_0^\ell p^{-a\ell} S_a(n,\ell)}
\prod_{k=1}^{l{}}\frac{\Gamma(n_k-a)}{\Gamma(1-a)}.
\end{align}
We define  the EPPF in (\ref{eq:EPPF})  
as the generalized Chinese restaurant sampling formula (gCRSF),  and we denote a random draw under this EPPF as 
$$\zv|n\sim{\mbox{gCRSF}}(n,\gamma_0,a,p).$$
The conditional distribution of the cluster number in a sample of size $n$ 
can be expressed as
\begin{align}\label{eq:f_L2_0}
p_L(l|n,\gamma_0,a,p) =\frac{1}{l!}\sum_{*}\frac{n!}{\prod_{k=1}^l n_k}p(\zv|n,\gamma_0,a,p)= \frac{\gamma_0^{l}  p^{-al}S_a(n,l)}{
\sum_{\ell=0}^n \gamma_0^\ell p^{-a\ell} S_a(n,\ell)}.
\end{align}

Note that if $a\rightarrow 0$, we  recover, from (\ref{eq:EPPF}), the Ewens sampling formula  which is the EPPF of the Chinese restaurant process (CRP) \citep{aldous:crp}. 
The prediction rule for  the EPPF in (\ref{eq:EPPF})   can be expressed as 
\beq\label{eq:PredictRule}
P(z_{i} = k|\zv^{-i},n,\gamma_0,a,p)  \propto
\begin{cases}
n_k^{-i} -a, &  {\mbox{for }} k=1,\ldots,l^{-i};\\ 
\gamma_0 p^{-a}, & {\mbox{if } }k=l^{-i}+1. 
\end{cases}
\eeq
This prediction rule can be used in a Gibbs sampler to simulate an exchangeable random partition $\zv|n\sim{\mbox{gCRSF}}(n,\gamma_0,a,p)$ of $[n]$. However, a large number of Gibbs sampling iterations may be required to generate an unbiased sample from this EPPF. Below we present a sequential construction for  this EPPF. 

Marginalizing out $z_n$ from (\ref{eq:EPPF}), we have
\begin{align}
{p(z_{1:n-1}|n,\gamma_0,a,p)}
~~=&~~~p(z_{1:n-1}|n-1,\gamma_0,a,p) \notag\\
&\times \frac{\sum_{\ell =0}^{n-1} \gamma_0^\ell p^{-a\ell}S_a(n-1,\ell)}{\sum_{\ell=0}^n \gamma_0^\ell p^{-a\ell }S_a(n,\ell)}\left[\gamma_0p^{-a}+ (n-1)- a l_{(n-1)}\right],\notag
\end{align}
where $z_{1:i}:=\{z_1,\ldots,z_i\}$, $l_{(i)}$ denotes the number of partitions in $z_{1:i}$ and $l_{(n)}=l$.
Further marginalizing out $z_{n-1},\ldots,z_{i+1}$, we have
\begin{align}
{p(z_{1:i}|n,\gamma_0,a,p)}
&=p(z_{1:i}|i,\gamma_0,a,p)\frac{\sum_{\ell=0}^{i} \gamma_0^\ell p^{-a\ell}S_a(i,\ell)}{\sum_{\ell=0}^n \gamma_0^\ell p^{-a\ell }S_a(n,\ell)}   R_{n,\gamma_0,a,p}(i,l_{(i)}) \notag\\
&= \frac{R_{n,\gamma_0,a,p}(i,l_{(i)}) \gamma_0^{l_{(i)}}  p^{-al_{(i)}}}{\sum_{\ell=0}^n \gamma_0^\ell p^{-a\ell}S_a(n,\ell)}  \prod_{k\,:\,n_{k,(i)}>0}\frac{\Gamma(n_{k,(i)}-a)}{\Gamma(1-a)} ,\label{eq:SizeEPPF}
\end{align}
where $n_{k,(i)}:=\sum_{j=1}^i \mathbf{1}(z_j=k)$; $R_{n,\gamma_0,a,p}(i,j)\equiv 1$ if $i=n$  and is recursively calculated for $i=n-1,m-2,\ldots,1$ with
\beq\label{eq:R}
R_{n,\gamma_0,a,p}(i,j) = R_{n,\gamma_0,a,p}(i+1,j)(i-a j) +  R_{n,\gamma_0,a,p}(i+1,j+1)\gamma_0p^{-a}.
\eeq
We name (\ref{eq:SizeEPPF}) as a size-dependent EPPF as its distribution on an exchangeable random partition of $[i]$ is a function of the sample size $n$.
Note that if $a=0$, then $$\frac{\sum_{l=0}^{i} \gamma_0^lp^{-al}S_a(i,l)}{\sum_{l=0}^n \gamma_0^lp^{-al}S_a(n,l)} = \frac{\sum_{l=0}^{i} \gamma_0^l |s(i,l)|}{\sum_{l=0}^n \gamma_0^l|s(n,l)|}  = \frac{\Gamma(i+\gamma_0)}{\Gamma(n+\gamma_0)}$$ and $R_{n,\gamma_0,a=0,p}(i,l) = \frac{\Gamma(n+\gamma_0)}{\Gamma(i+\gamma_0)}$, and hence ${p(z_{1:i}|n,\gamma_0,a=0,p)}
\equiv p(z_{1:i}|i,\gamma_0,a=0,p)$. Thus when $a=0$, the EPPF becomes independent of the sample size,  which is a well-known property for the Chinese restaurant process.

\begin{cor}[Sequential Construction]
Since $p(z_{i+1} |z_{1:i},n,\gamma_0,a,p)  = \frac{p(z_{1:i+1} | n,\gamma_0,a,p)} {p(z_{1:i} | n,\gamma_0,a,p)}$, 
conditioning on the sample size $n$, the sequential prediction rule of the generalized Chinese restaurant sampling formula $\zv|n\sim\emph{\mbox{gCRSF}}(n,\gamma_0,a,p) $ can be expressed as
 \beq\label{eq:PredictRulej}
P(z_{i+1} = k|z_{1:i},n,\gamma_0,a,p) =
\begin{cases}\vspace{3mm}
(n_{k,(i)} -a) \frac{R_{n,\gamma_0,a,p}(i+1, ~l_{(i)})}{R_{n,\gamma_0,a,p}(i,~ l_{(i)})}, &  {\mbox{for }} k=1,\ldots,l_{(i)};\\  \gamma_0 p^{-a}\frac{R_{n,\gamma_0,a,p}(i+1, ~l_{(i)}+1)}{R_{n,\gamma_0,a,p}(i, ~l_{(i)})}, & {\mbox{if } }k=l_{(i)}+1; \end{cases}
\eeq
where $i=1,\ldots,n-1$.\end{cor}

With this sequential prediction rule, similar to an EPPF of $\Pi$, we can construct $\Pi_{i+1}$ from $\Pi_i$ in a sample of size $n$ by assigning element $(i+1)$ to $A_{z_{i+1}}$.
 When $a=0$, we have $$\frac{R_{n,\gamma_0,a,p}(i+1, ~l_{(i)})}{R_{n,\gamma_0,a,p}(i,~ l_{(i)})} = \frac{R_{n,\gamma_0,a,p}(i+1, ~l_{(i)}+1)}{R_{n,\gamma_0,a,p}(i,~ l_{(i)})} = \frac{\Gamma(i+\gamma_0)}{\Gamma(i+1+\gamma_0)} =  \frac{1}{i+\gamma_0},$$  and this sequential prediction rule 
becomes the same as that of a Chinese restaurant process with concentration parameter $\gamma_0$.

\begin{cor} 
The distribution of the number of clusters in $z_{1:i}$ in a sample of size $n$ can be expressed as
\begin{align}\label{eq:l_i}
{p(l_{(i)}|n,\gamma_0,a,p)}
&=p(l_{(i)}|i,\gamma_0,a,p)\frac{\sum_{\ell=0}^{i} \gamma_0^\ell p^{-a\ell} S_a(i,\ell)}{\sum_{\ell=0}^n \gamma_0^\ell p^{-a\ell}S_a(n,\ell)}   R_{n,\gamma_0,a,p}(i,l_{(i)}),\notag\\
&=\frac{ \gamma_0^{l_{(i)}}p^{-al_{(i)}}S_a(i,l_{(i)})R_{n,\gamma_0,a,p}(i,l_{(i)})}{\sum_{\ell=0}^n \gamma_0^\ell p^{-a\ell}S_a(n,\ell)}.
\end{align}
\end{cor}
 This can be directly derived using  (\ref{eq:SizeEPPF}) and the relationship between the EPPF and the distribution of the number of clusters. From this PMF, we obtain a useful identity
\beq\notag
 {\sum_{\ell=0}^n \gamma_0^\ell p^{-a\ell}S_a(n,\ell)} = \gamma_0p^{-a}R_{n,\gamma_0,a,p}(1,1),
\eeq
which could be used to calculate the PMF of the generalized NB distribution in (\ref{eq:f_M0}) and the EPPF in (\ref{eq:EPPF}) without the need to compute the generalized Stirling numbers $S_a(n,l)$. 

\begin{cor} Given the model parameters $\gamma_0$, $a$ and $p$, the probability for two elements uniformly at random selected from a random sample of size $n$ to be in two different  groups   
can be expressed as
\beq\label{eq:z12}
P(z_1\neq z_2|n,\gamma_0,a,p) = \frac{\gamma_0 p^{-a}R_{n,\gamma_0,a,p}(2,2)}{R_{n,\gamma_0,a,p}(1,1)}= \left[1+ \frac{1-a}{\gamma_0p^{-a}}\frac{R_{n,\gamma_0,a,p}(2,1)}{R_{n,\gamma_0,a,p}(2,2)}\right]^{-1}.
\eeq
When $a=0$, for $n\ge 2$, we have $$P(z_1\neq z_2|n,\gamma_0,a=0,p)\equiv\frac{\gamma_0}{1+\gamma_0}.$$
\end{cor}

\begin{proof} We directly obtain (\ref{eq:z12}) by setting $i=1$ and $z_{i+1}=2$ in (\ref{eq:PredictRulej})  and using the recursive definition of 
$R_{n,\gamma_0,a,p}(1,1)$ in (\ref{eq:R}).
 \end{proof}

 \begin{cor}[Simpson's Index of Diversity] Given the model parameters $\thetav=\{\gamma_0, a, p\}$, the probability for two individuals uniformly at random selected from a random sample, whose size follows $n\sim\emph{\mbox{gNB}}(\gamma_0,a,p)$ and is larger than two, to be in two different  groups  
can be expressed as
\begin{align}\label{eq:z12_random}
S_{\thetav} &:=  P(z_1\neq z_2|\gamma_0,a,p)  = \sum_{n=2}^\infty P(z_1\neq z_2|n,\gamma_0,a,p) \frac{\emph{\mbox{gNB}}(n;\gamma_0,a,p)}{1-\emph{\mbox{gNB}}(0;\gamma_0,a,p)-\emph{\mbox{gNB}}(1;\gamma_0,a,p)}\notag\\
&~ = \frac{\gamma_0^2 p^{-2a} e^{-{\gamma_0}\frac{1-(1-p)^a}{ap^a}}}{1-e^{-{\gamma_0}\frac{1-(1-p)^a}{ap^a}}-\gamma_0 p^{1-a}e^{-{\gamma_0}\frac{1-(1-p)^a}{ap^a}}} \sum_{n=2}^\infty  \frac{p^n}{n!} R_{n,\gamma_0,a,p}(2,2).
\end{align}
When $a=0$, we have $$P(z_1\neq z_2|\gamma_0,a=0,p)\equiv\frac{\gamma_0}{1+\gamma_0}.$$

\end{cor}

Under this construction, 
given a random species sample $(z_1,\ldots,z_{n})$, with a prior distribution on $\thetav$ as $p_\Theta(\thetav)$, the posterior mean of Simpson's index of diversity is expressed as
\beq
S 
=  \int 
S_{\thetav} p(\thetav|z_1,\ldots,z_n) d\thetav,
\eeq
where $$ p(\thetav|z_1,\ldots,z_n) = \frac{p(z_1,\ldots,z_n, n|\thetav)p_\Theta(\thetav) }{\int p(z_1,\ldots,z_n, n|\thetav)p_\Theta(\thetav) d\thetav}.$$
In the next section we show how to peform MCMC estimation for the model from which we will derive the posterior value for Simpson's index of diversity.

\section{Illustrations}\label{sec:results}

Species abundance data of a sample is usually represented with a set of frequency counts $M=\{m_1,m_2,\ldots\}$, where $m_i$ denotes the number of species that have been observed $i$ times in the sample. 
This data can also be converted into a sequence of group indices $\zv=(z_1,\ldots,z_n)$ or a group-size vector $(n_1,\ldots,n_l)$, where $n_k$ is the number of individuals in group $k$, $n = \sum_i  im_i=\sum_{k=1}^l n_k$ is the size of the sample and $l=\sum_i m_i$ is the number of distinct groups in the sample. For example, we may represent $M=\{m_1=2, m_2=1, m_3=2\}$ as $\zv=(1,2,3,3,4,4,4,5,5,5)$ or $(n_1,\ldots,n_5)=(1,1,2,3,3)$. For a sample of species frequency counts, we use (\ref{eq:f_Z_M})  as the likelihood for the model parameters $\thetav=\{\gamma_0,a,p\}$. With appropriate priors imposed on $\thetav$, we use MCMC to obtain posterior samples $\thetav^{(j)}=\{\gamma_0^{(j)}, a^{(j)}, p^{(j)}\}$ and then calculate $S_{\thetav^{(j)}}$.
The details of MCMC update equations are provided in the Appendix.  

\subsection{Estimation of T-cell Receptor Diversity}\label{sec:TCR}
An important characteristic of the immune system is the diversity of T-cell receptors (TCRs) \citep{nikolich2004many,ferreira2009non}. As the number of distinct TCRs might be extremely high in the body, one usually investigates TCR diversity by collecting a sample of T-cells and determining the number of distinct TCR sequences and their respective abundances (counts) in that sample. For example,  a  Bayesian semiparametric approach is proposed in \citet{guindani2014bayesian}  to estimate TCR diversity of regulatory, Treg, and conventional T-cells, Tconv, across samples of two healthy and three diabetic mice; the TCR diversity there is defined as the number of distinct TCR sequences in a sample, including $k'$ observed  distinct TCR sequences and $k_0$ unobserved ones due to censoring of zero counts. In this paper, we estimate TCR diversity by calculating Simpson's index of diversity given a sample of species frequency counts.

 \begin{figure}[!tb]
 
 \begin{center}
\includegraphics[width=0.85\columnwidth]{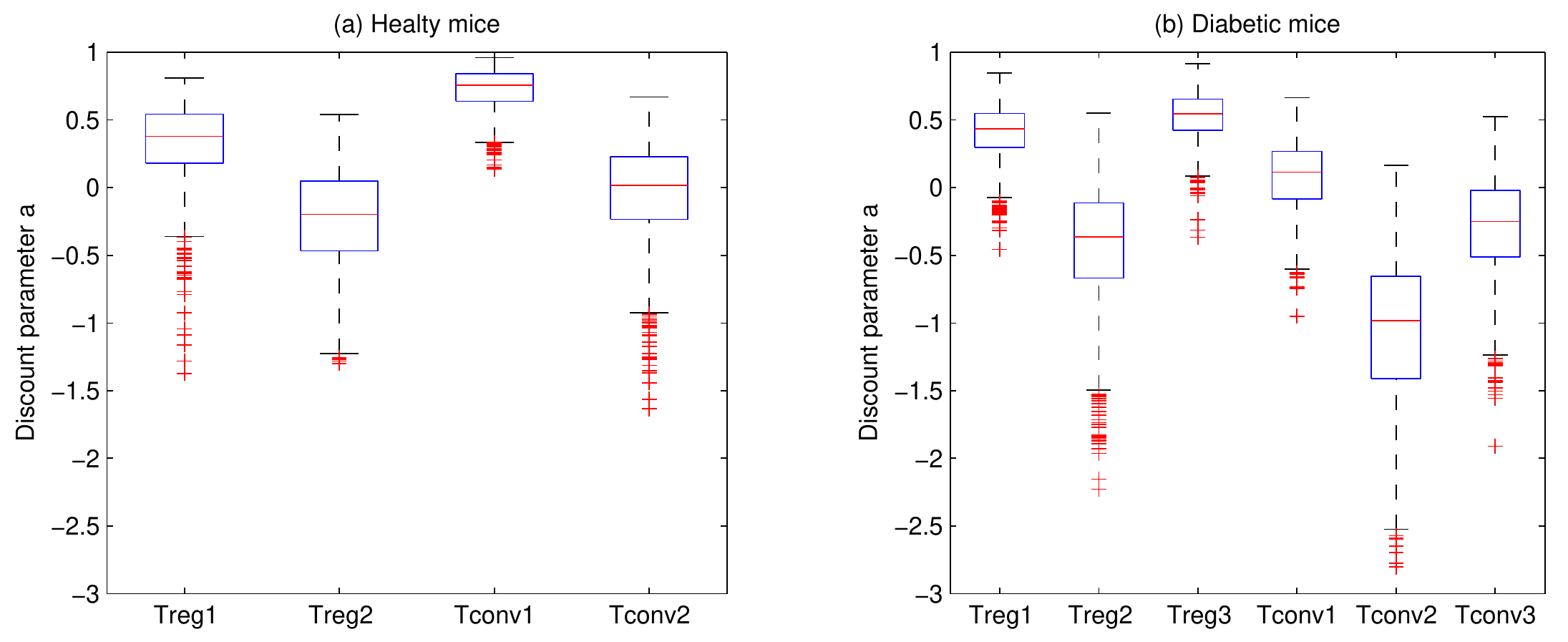}
\end{center}
\vspace{-6mm}
\caption{ \label{fig:Tcell_a} \small Box plots of $\{a^{(j)}\}_{j=1,N}$, the posterior MCMC samples of the discount parameter $a$, for regulatory, Treg, and conventional T-cells, Tconv, across various samples of (a) two healthy and (b) three diabetic mice.
}
 
\begin{center}
\includegraphics[width=0.85\columnwidth]{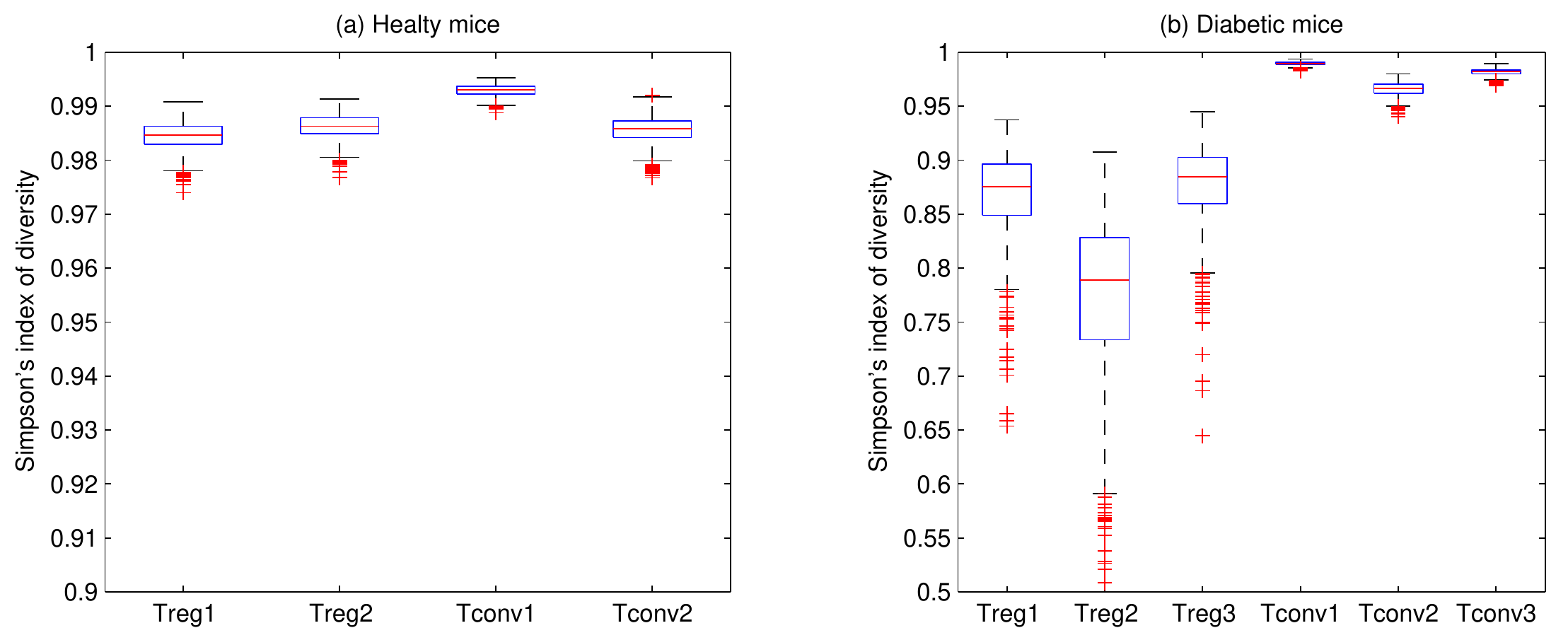}
\end{center}
\vspace{-6mm}
\caption{ \label{fig:Tcell} \small Box plots of $\{S_{\thetav^{(j)}}\}_{j=1,N}$, the posterior MCMC samples of Simpson's index of diversity, for regulatory, Treg, and conventional T-cells, Tconv, across various samples of (a) two healthy and (b) three diabetic mice. 
}


\end{figure}

Considering the same TCR species abundance frequency count dataset used in \citet{ferreira2009non} and presented in Table 2 of \citet{guindani2014bayesian},  we compare Simpson's indice of diversity of the TCRs of Treg and Tconv across samples of two healthy and three diabetic mice. For example, for Treg, we have 
$M=\{40, 5, 5, 2, 3\}$ with $i\in\{1,2,3,4,5\}$  for the sample  from heathy mouse 1, and we have  $M =\{ 8, 1, 2,  1, 1, 1\}$ with $i\in\{1,2,3,5,36,40\}$ for the sample from diabetic mouse 1. For each sample of T-cells, we consider 2000 MCMC iterations and collect the last 1000 MCMC samples $\{\thetav^{(j)}\}_{1,1000}$.

Figure  \ref{fig:Tcell_a} shows the box plots of the MCMC posterior samples of the discount parameter~$a$ in various samples of regulatory and conventional T-cells for the healthy and diabetic mice. We find no clear associations between the posteriors  of $a$ and whether the mice are healthy or diabetic or whether the T-cells are regulatory or conventional. 

As shown in Figure  \ref{fig:Tcell},  using the samples for the diabetic mice, the estimated Simpson's indices of diversity of the TCRs for regulatory T-cells are considerably lower than those for conventional T-cells; whereas for  the healthy mice, no clear differences on TCR diversity are found. 
Comparing Figures \ref{fig:Tcell_a} and \ref{fig:Tcell}, one may also not find clear relationships between the estimated values of $a$ and the estimated Simpson's indices of diversity, which suggests that for the generalized negative binomial process, the discount parameter $a$ alone may not be a good indicator for species evenness measured by Simpson's index of diversity. \citet{guindani2014bayesian} showed that diabetic mice tended to have a smaller number of distinct TCRs in a sample of regulatory T-cells than in a sample of conventional T-cells. Our comparison of Simpson's indices of diversity, which measure species evenness and hence complementary to the  comparison of species richness studied in
\citet{guindani2014bayesian}, provides additional evidence to suggest that for diabetic mice, the TCR diversity of regulatory T-cells is lower than that of conventional T-cells.

\subsection{Genomic Data Analysis}

An important research topic in genomics is the analysis of expressed sequence tag (EST) data, which arise by  sequencing complementary DNA (cDNA) libraries consisting of millions of genes. The number of ESTs from a particular gene indicates  the expression level of that gene. It is typical that  only a small portion of the cDNA is  sequenced in a sample due to cost constrains, and one need to rely on this sample to estimate population properties. 
We consider a tomato flower EST dataset, previously analyzed in  \citet{mao2002poisson} and \citet{lijoi2007bayesian}, that consists of 2586 ESTs from 1825 genes as $M= \{1434, 253, 71, 33, 11, 6, 2, 3, 1, 2, 2, 1, 1, 1, 2, 1, 1\}$ for 
       $i \in \{1,\ldots,14\} \bigcup\{16, 23, 27\}$. We convert $\{m_i\}_i$ into $(z_1,\ldots,z_{2586})$.
To evaluate the accuracy of the proposed nonparametric Bayesian estimator in (\ref{eq:z12_random}),  we consider this relatively large sample as the population and treat   $\widehat{S}=0.9993$, a sample estimate with (\ref{eq:S_hat}), as the ``true'' Simpson's index of diversity for the population.

\begin{table}[!t]
\caption{Simulation study based on 100 expressed sequence tag (EST) samples of size $50$ uniformly at random selected from a population of 2586 ESTs from 1825 distinct genes, with various settings of the discount parameter $a$. A sample estimate of 0.9993 using all the 2586 ESTs  is considered as the ``true'' Simpson's index of diversity for  the population.}
\begin{center}\label{tab:simulation}
\begin{tabular}{|c|c|c|c|c|}
\hline
Parameter Setting & Mean Bias  & Median Bias &  50\% Coverage  & 95\% Coverage \\
 & ($\times 10^{-3}$) & ($\times 10^{-3}$) &    & \\
\hline
$a=-1$ & 10.37 & 10.60 &0\% & 0\%\\
$a=0$ &3.05 & 3.31 & 0\% & 0\%\\
$a=0.5$ &1.07 & 1.40  &18\% & 85\%\\ 
$a<0$ & 3.51  & 3.78 &  0\% & 0\% \\
$0\le a<1$ &0.48 & 1.11& 62\% & 98\%\\
$a<1$ & \textbf{0.41}  & \textbf{1.09}& \textbf{69\%} & \textbf{99\%}\\
\hline
\end{tabular}
\end{center}
\label{default}
\end{table}%

        We randomly select an EST sample of size $n=50$ from $(z_1,\ldots,z_{2586})$ to estimate the Simpson's index of diversity of the population. 
       For each selected EST sample,   we use MCMC to obtain posterior samples $\thetav^{(j)}=\{\gamma_0^{(j)}, a^{(j)}, p^{(j)}\}$ and then calculate $S_{\thetav^{(j)}}$; we consider 2000 MCMC iterations and collect one sample in every five iterations in the last 1000 MCMC iterations, leading to $N=200$ total samples $\{\thetav^{(j)}\}_{1,200}$; we find from the collected MCMC samples the mean, 
       median, 50 percentile range and 95 percentile range of $\{S_{\thetav^{(j)}}\}$, and compare these values against $0.9993$.    We repeat the same procedure 100 times and find the averages among these 100 times of the absolute distances from the mean and median to 0.9993, and the probabilities for 0.9993 to be covered by the 50 and 95 percentile ranges. 
       
       We summarize the results in Table \ref{tab:simulation}, where we fix $a$ to be $-1$, $0$ or $ 0.5$, or let $a$ be inferred for each EST sample and restrict it to be $a<0$, $0\le a<1$ or $a<1$. 
       It is clear that allowing $a$ to be freely adjusted within $(-\infty,1)$ leads to a more accurate estimation of Simpson's index of diversity using a sample of the population, demonstrating the effectiveness of the generalized negative binomial process on the analysis of EST sequence counts. Similar simulation results are observed on the TCR sequence count dataset studied in Section \ref{sec:TCR}.

In conclusion, we have introduced a sample size dependent species model, which allows flexible modeling of species abundance frequency count data. We gain this flexibility with a simple model and consequently posterior inference via MCMC is also simple. The paper provides a framework to jointly model a single random count and its exchangeable random partition. 
It is natural to extend the same framework  to mixture modeling, where the usual task is to partition a set of data points into exchangeable clusters, where both the number and sizes of clusters are unknown and need to be inferred.  The techniques developed here to model a random count vector also serve as the foundation for  \citet{NBP_CountMatrix} to construct a family of nonparametric Bayesian priors for infinite random count matrices, and for \citet{BNBP_EPPF} to define a prior distribution that describes the random partition of a count vector into a latent random count matrix.

%

\small
\bibliographystyle{plainnat}
\bibliography{References072014}
\normalsize

\appendix
\section{Proof for Theorem \ref{thm:compoundPoisson}}
\begin{proof}
Let us consider the process $X_G$, conditional on $G$, given by
$$X_G(A)=\sum_{k} n_k\,{\bf 1}(\omega_k\in A).$$  
Now it is easy to see that
$$\E[\exp\{-\phi X_G(A)\}|G]=\exp\{-G(A)(1-e^{-\phi})\},$$
and using the well known result for homogeneous  L\'evy processes, we have
\beq
\E[\exp\{-\lambda G(A)\}]=\exp\left\{-G_0(A)\,\int_0^\infty \left[1-e^{-\lambda s}\right]\,\rho(\d s)\right\}.\label{one}
\eeq
Now, the key observation is the following identity:
$$1-e^{-(1-e^{-\phi})s}=1-e^{-s}\sum_{j=0}^\infty \frac{s^j}{j!}e^{-\phi j}=(1-e^{-s})-e^{-s}\sum_{j=1}^\infty \frac{s^j}{j!}e^{-\phi j}.$$
Let us put this to one side for now and consider the model for $\tilde{X}$ given by
$$\tilde{X}(A)=\sum_{k=1}^{l} n_k\,{\bf 1}(\omega_k\in A)$$
with $l\sim\mbox{Po}(\gamma G_0(\Omega))$ for some non-negative $\gamma$ and 
independently $P(n_k=j)=\pi_j$ for some $\pi_j\leq 1$ and $j\in\{1,2,\ldots\}$.
Now given $l$, we have
$$\E[ \exp\{-\phi \tilde{X}(A)\}|l]=\prod_{k=1}^{l} \E [\exp\{-\phi n_k\,{\bf 1}(\omega_k\in A)\}]$$
and each of these expectations is given by
$$\psi=\sum_{j=1}^\infty e^{-\phi j}\pi_j.$$
Thus
$$\E[\exp\{-\phi \tilde{X}(A)\}]=
\exp\{-\gamma\, G_0(A)\, (1-\psi)\}$$
which is given by
\beq
\exp\left\{-\gamma\,G_0(A)\, \left[1-\sum_{j=1}^\infty e^{-\phi j}\,\pi_j\right]\right\}.\label{two}
\eeq
Comparing (\ref{one}) and (\ref{two}) we see that we have a match when
$$\gamma=\int_0^\infty (1-e^{-s})\,\rho(\d s)$$
and
$$\pi_j=\frac{\int_0^\infty s^j\,e^{-s}\,\rho(\d s)}{j! \gamma}\,,$$
and note that it is easy to verfy that
$$\sum_{j=1}^\infty \pi_j=1.$$
\end{proof}

\section{Proof for Corollary \ref{thm:predict}}

This follows directly  from Bayes' rule, since $p(z_i|\zv^{-i},n,\gamma_0,\rho) = \frac{p(z_i,\zv^{-i},n|\gamma_0,\rho)}{p(\zv^{-i},n|\gamma_0,\rho)}$, where 
$$p(z_i,\zv^{-i},n|\gamma_0,\rho)=\,\,\,\,\,\,\,\,\,\,\,\,\,\,\,\,\,\,\,\,\,\,\,\,\,\,\,\,\,\,\,\,\,\,\,\,\,\,\,\,\,\,\,\,\,\,\,\,\,\,\,\,\,\hfill$$
$$n^{-1\,} p(\zv^{-i},n-1|\gamma_0,\rho)\,\left[\gamma_0\int_0^\infty se^{-s}\rho(\d s)\,{\bf 1}(z_i=l^{-i}+1) \,+ \,\sum_{k=1}^{l^{-i}} \frac{\int_0^\infty s^{n_k^{-i}+1} e^{-s} \rho(\d s)}{\int_0^\infty s^{n_k^{-i}} e^{-s} \rho(\d s)} {\bf 1}(z_i=k)  \right].\notag$$
Marginalizing out the $z_i$ from $p(z_i,\zv^{-i},n|\gamma_0,\rho)$ we have
\beqs
&p(\zv^{-i},n|\gamma_0,\rho)=
n^{-1}\,p(\zv^{-i},n-1|\gamma_0,\rho)\left[{\gamma_0\int_0^\infty se^{-s}\rho(\d s)+ \sum_{k=1}^{l^{-i}} \frac{\int_0^\infty s^{n_k^{-i}+1} e^{-s} \rho(\d s)}{\int_0^\infty s^{n_k^{-i}} e^{-s} \rho(\d s)}  }\right].\notag
\eeqs

\section{Derivations for the GNBP}

Marginalizing out $\lambda$ from  
$[n| \lambda]\sim\mbox{Po}(\lambda)$ with $\lambda\sim{{}}\mbox{gGamma}[\gamma_0,a,p/(1-p)]$, leads to a generalized NB distribution; $n\sim\mbox{gNB}(\gamma_0,a,p)$, with shape parameter $\gamma_0$, discount parameter $a<1$, and probability parameter $p$. The probability generating function (PGF) is given by
$$ \E[t^n]   = \E[\E[t^n|\lambda]] 
 = \exp\left\{-\frac{\gamma_0[(1-pt)^a-(1-p)^a)]}{ap^a}\right\}, 
$$
the mean value is $\gamma_0\big[p/(1-p)\big]^{1-a}$ and the variance is $\gamma_0\big[p/(1-p)\big]^{1-a}(1-ap)/(1-p)$. The PGF  was originally presented in \citet{willmot1988remark} and \citet{gerber1992generalized}. With the PGF written  as 
$$\begin{array}{ll}
\E(t^n)  & =\exp\left\{\gamma_0\frac{(1-p)^a}{ap^a}\right\}\sum_{k=0}^\infty \frac{1}{k!} {\left(\frac{-\gamma_0(1-pt)^a}{ap^a}\right)^k}  \\ \\
& =\exp\left\{\gamma_0\frac{(1-p)^a}{ap^a}\right\} \sum_{k=0}^\infty \frac{1}{k!} {\left(\frac{-\gamma_0}{ap^a}\right)^k} \sum_{j=0}^\infty \binom{ak}{j}(-pt)^j,\end{array}$$
we can derive the PMF as 
\beqs\label{eq:f_M}
p_N(n|\gamma_0,a,p)
 = 
\frac{p^n}{n!}e^{{\gamma_0}\frac{(1-p)^a}{ap^a}} \sum_{k=0}^\infty \frac{1}{k!}{\left(-\frac{\gamma_0}{ap^a}\right)^k} \frac{\Gamma(n-ak)}{\Gamma(-ak)}, ~n=0,1,\ldots.
\eeqs 
We can also generate 
$n\sim{{}}\mbox{gNB}(\gamma_0,a,p)$ 
from a compound Poisson distribution, as
$
n=\sum_{k=1}^l n_k$, with the  $(n_k)$ independent from $\mbox{TNB}(a,p)$, and $l\sim\mbox{Po}\big(\frac{\gamma_0(1-(1-p)^a)}{ap^a}\big),
$
where $\mbox{TNB}(a,p)$ denotes a truncated NB distribution, with PGF $\E[t^{u}] =  
\frac{1-(1-pt)^a}{1-(1-p)^a}$ and
PMF
\beqs\label{eq:TNB}
p_U(u|a,p)= \frac{\Gamma(u-a)}{u!\Gamma(-a)}\frac{p^u(1-p)^{-a}}{1-(1-p)^{-a}},~u=1,2,\ldots.
\eeqs
Note that as $a\rightarrow 0$, $u\sim\mbox{TNB}(a,p)$ becomes a logarithmic distribution \citep{LogPoisNB} with PMF $p_U(u|p)=\frac{-1}{\ln(1-p)}\frac{p^u}{u}$ 
and $n\sim\mbox{gNB}(\gamma_0,a,p)$ becomes a NB distribution; $n\sim\mbox{NB}(\gamma_0,p)$. The truncated NB distribution with $0<a<1$ is the extended NB distribution introduced in \citet{engen1974species}. 

Here we provide a useful identity which we will be used later in this section.
Denote by $\sum_{*}$ as the summation over all sets of positive integers $(n_1,\ldots,n_l)$ with ${\sum_{k=1}^l n_k = n}$.  We call $n\sim\mbox{SumTNB}(l,a,p)$ as a sum-truncated NB distributed random variable that can be generated via $n=\sum_{k=1}^l n_k, ~n_k\sim\mbox{TNB}(a,p)$. Using both (\ref{eq:TNB}) 
and $$
\left[\frac{1-(1-pt)^a}{1-(1-p)^a}\right]^{l}= \frac{ \sum_{k=0}^l \binom{l}{k} (-1)^k \sum_{j=0}^\infty\binom{ak}{j} (-pt)^j }{  [1-(1-p)^a]^l},
$$  we may express the PMF of the sum-truncated NB distribution 
as
$$p_N(n|l,a,p) = 
\sum_{*} \prod_{k=1}^l {\frac{\Gamma(n_k-a)}{n_k!\Gamma(-a)} \frac{p^{n_k}(1-p)^{-a}}{1-(1-p)^{-a}}}
=\frac{p^n}{ [1-(1-p)^a]^l} {\sum_{k=0}^l (-1)^k \binom{l}{k}  \frac{\Gamma(n-ak)}{n!\Gamma(-ak)}  },
$$ 
leading to the identity shown  in (\ref{eq:gStirling}).

The EPPF 
is the ECPF in (\ref{eq:f_Z_M}) divided by the marginal distribution of $n$ in (\ref{eq:f_M}), given by
\begin{align} \label{eq:EPPF1}
 p(\zv|n,\gamma_0,a,p) &= 
 p_n(z_1,\ldots,z_n|n) 
=\frac{e^{-\frac{\gamma_0}{ap^{a}}} }{ \sum_{k=0}^\infty \frac{1}{k!}{\left(-\frac{\gamma_0}{ap^a}\right)^k}\frac{\Gamma(n-ak)}{\Gamma(-ak)} }  \gamma_0^{l{}}  p^{-al{}} 
\prod_{k=1}^{l{}}\frac{\Gamma(n_k-a)}{\Gamma(1-a)}.
\end{align}

Using 
the EPPF  in (\ref{eq:EPPF}) and the identity in (\ref{eq:gStirling}), the conditional distribution of the number of clusters $l$ in a sample of size $n$ 
can be expressed as
\begin{align}\label{eq:f_L2}
p_L(l|n,\gamma_0,a,p)& = \frac{1}{l!}\sum_{*}\frac{n!}{\prod_{k=1}^l n_k!}p(\zv|n,\gamma_0,a,p)
 =\frac{ \gamma_0^l p^{-al} S_a(n,l)}{e^{\frac{\gamma_0}{ap^{a}}}   \sum_{k=0}^\infty \frac{1}{k!} {\left(\frac{-\gamma_0}{ap^a}\right)^k}  \frac{\Gamma(n-ak)}{\Gamma(-ak)}} ,
\end{align}
which, since $\sum_{l=0}^n p_L(l|n,\gamma_0,a,p)=1$, further leads to identity 
\beq
e^{\frac{\gamma_0}{ap^{a}}}   \sum_{k=0}^\infty \frac{1}{k!} {\left(\frac{-\gamma_0}{ap^a}\right)^k}  \frac{\Gamma(n-ak)}{\Gamma(-ak)} = \sum_{l=0}^n \gamma_0^l p^{-al} S_a(n,l). \notag
\eeq
Applying this identity on (\ref{eq:f_M}), (\ref{eq:EPPF1}) and (\ref{eq:f_L2}) lead to (\ref{eq:f_M0}),
(\ref{eq:EPPF})  and (\ref{eq:f_L2_0}).

\section{MCMC Inference}

For the gNBP, the ECPF in (\ref{eq:f_Z_M}) defines a fully factorized likelihood for $\gamma_0$, $a$ and $p$. With a gamma prior $\mbox{Gamma}(e_0,1/f_0)$ placed on $\gamma_0$, we have 
\beqs
(\gamma_0|-)\sim\mbox{Gamma}\bigg(e_0 + l{},\frac{1}{f_0+ \frac{1-(1-p)^a}{ap^a}}\bigg).
\eeqs
As $a\rightarrow 0$, we have
$
(\gamma_0|-)\sim\mbox{Gamma}\left(e_0 + l{},\frac{1}{f_0- \ln(1-p)}\right). 
$ This paper sets $e_0=f_0=0.01$.

Since $a<1$, we have $\tilde{a}=\frac{1}{1+(1-a)} \in(0,1)$. With a uniform prior placed on $\tilde{a}$ in $(0,1)$ and the likelihood of gNBP in (\ref{eq:f_Z_M}), we use the  griddy-Gibbs sampler \citep{griddygibbs} to
sample $a$ from a discrete distribution
\beq
P(a|-)\propto e^{-\gamma_0\frac{1-(1-p)^a}{ap^{a}}}
p^{-al{}}
\prod_{k=1}^{l{}} \frac{\Gamma(n_k-a)}{{\Gamma(1-a)} }
\eeq
over a grid of points $\frac{1}{1+(1-a)}=0.0001,0.0002,\ldots,0.9999$.

We place a uniform prior on $p$ in $(0,1)$.
When $a\rightarrow 0$, the likelihood of the gNBP in (\ref{eq:f_Z_M}) becomes proportional to $p^{m}(1-p)^{\gamma_0}$, thus we have 
$
(p|-)\sim\mbox{Beta}(1+n,1+\gamma_0).
$
When $a\neq 0$,  
we use the  griddy-Gibbs sampler to sample $p$ from a discrete distribution
\beq
P(p|-)\propto e^{-\gamma_0\frac{1-(1-p)^a}{ap^{a}}}
 p^{n-al{}}
\eeq over 
a grid of points $p=0.001,0.002,\ldots,0.999$.

\end{document}